\documentclass[12pt,reqno,a4paper]{amsart}
\usepackage{amsmath,amssymb,dsfont,graphicx,cite,latexsym,
epsf,cancel,epic,eepic}
\usepackage{amssymb,mathrsfs}
\usepackage[colorlinks=false]{hyperref}
\usepackage{mathpazo}
\newtheorem{theorem}{Theorem}

\newtheorem{proposition}[theorem]{Proposition}

\newtheorem{lemma}{Lemma}

\newtheorem{definition}{Definition}

\def\beq{\begin{equation}}
\def\eeq{\end{equation}}
\def\bea{\begin{eqnarray}}
\def\eea{\end{eqnarray}}
\def\benpf{\noindent {\textbf{{\emph{Proof.}}\;}}}
\def\endpf{\hfill$\blacksquare$\medskip}
\makeatletter
\expandafter\let\expandafter
\reset@font\csname reset@font\endcsname
\def\subeqnarray{\arraycolsep1pt
   \def\@eqnnum\stepcounter##1{\stepcounter{subequation}
       {\reset@font\rm(\theequation\alph{subequation})}}
\jot5mm     \eqnarray}

\makeatother

\newcommand{\R}{{\mathbb R}}
\newcommand{\D}{{\mathsf D}}
\newcommand{\dd}{{\mathsf d}}

\newcommand{\EE}{{\mathscr E}}

\def\epsilon{\varepsilon}

\def\pa{\partial}

\def\nn{\nonumber}
\def\endpf{\hfill$\square$\medskip}

\def\E{{\rm{e}}}

\newbox\meibox
\def\placeunder#1#2#3#4{\setbox\meibox%
\vbox{\hbox{\hskip#4$\hphantom{#2}$}\hbox{$\hphantom{#1}$}}%
\vtop{\baselineskip=0pt\lineskiplimit=\baselineskip%
\lineskip=#3\hbox to \wd\meibox{\hfil\hskip#4$#2$\hfil}%
\hbox to \wd\meibox{\hfil$#1$\hfil}}}
\def\intprod{\mathbin{\hbox to 6pt{%
                 \vrule height0.4pt width5pt depth0pt
                 \kern-.4pt
                 \vrule height6pt width0.4pt depth0pt\hss}}}

\begin{document}
\title[Variational symmetries and pluri-Lagrangian systems]
{Variational symmetries and pluri-Lagrangian systems\\ in classical mechanics}

\author{Matteo Petrera \and Yuri B. Suris }

\thanks{E-mail: {\tt  petrera@math.tu-berlin.de, suris@math.tu-berlin.de}}

\maketitle

\begin{center}
{\footnotesize{
Institut f\"ur Mathematik, MA 7-1\\
Technische Universit\"at Berlin, Str. des 17. Juni 136,
10623 Berlin, Germany
}}
\end{center}

\begin{abstract} 
We analyze the relation of the notion of a pluri-Lagrangian system, which recently emerged in the theory of integrable systems, to the classical notion of variational symmetry, due to E. Noether.  We treat classical mechanical systems and show that, for any Lagrangian system with $m$ commuting variational symmetries, one can construct a pluri-Lagrangian 1-form in the $(m+1)$-dimensional time, whose multi-time Euler-Lagrange equations coincide with the original system supplied with $m$ commuting evolutionary flows corresponding to the variational symmetries. We also give a Hamiltonian counterpart of this construction, leading, for any system of commuting Hamiltonian flows, to a pluri-Lagrangian 1-form with coefficients depending on functions in the phase space.

\end{abstract}


\section{Introduction}

This paper investigates some aspects of variational structure of integrable systems. The corresponding theory was initiated in \cite{LN1}, where it was shown that solutions of integrable (multi-dimensionally consistent) quad-equations on any quad-surface $\Sigma$ in $\mathbb{Z}^m$ are critical points of a certain action functional $\int_\Sigma\mathcal L$ obtained by integration of a suitable discrete Lagrangian 2-form $\mathcal L$. Moreover, it was observed in \cite{LN1} that the critical value of the action remains invariant under local changes of the underlying quad-surface, or, in other words, that the 2-form $\mathcal L$ is closed on solutions of quad-equations, and it was suggested to consider this as a defining feature of integrability. 

This line of research was developed in several directions, mainly by two teams: by Nijhoff with collaborators (under the name ``theory of Lagrangian multi-forms''), see \cite{LN1, LN2, LNQ, YLN, ALN}, and by the authors of the present paper with collaborators, who termed the corresponding structures ``pluri-Lagrangian'', see \cite{BS, BPS1,BPS2,BPS3,S1,S2,S3,SV}. As argued in \cite{BPS1}, the unconventional idea to consider the action on arbitrary two-dimensional surfaces in the multi-dimensional space of independent variables has significant precursors. These include:
\begin{itemize}
\item {\em Theory of pluriharmonic functions}. By definition, a pluriharmonic function of $m$ complex variables $f:\mathbb{C}^m\to\mathbb{R}$ minimizes the Dirichlet energy 
$E_\Gamma$ along any holomorphic curve in its domain $\Gamma:\mathbb{C}\to\mathbb{C}^m$, defined as
$$E_\Gamma=\int_\Gamma |(f\circ \Gamma)_z|^2dz\wedge d\bar z.$$ The similarity of this definition with the above mentioned property of the Lagrangian structure of multi-dimensionally consistent discrete systems motivates the term \emph{pluri-Lagrangian systems}, which was proposed in \cite{BPS1, BS}.
Differential equations governing pluriharmonic functions,
\[
\frac{\partial^2 f}{\partial z_i\partial \bar{z}_j}=0 \quad {\rm for\;\;all}\quad i,j=1,\ldots,m,
\]
are heavily overdetermined. Therefore it is not surprising that pluriharmonic functions belong to the theory of integrable systems.

\item \emph{Baxter's Z-invariance} of solvable models of statistical mechanics \cite{Bax1, Bax2}. This concept is based on invariance of the partition functions of solvable models under elementary local transformations of the underlying planar graphs. It is well known that one can identify planar graphs underlying these models with quad-surfaces in $\mathbb{Z}^m$. On the other hand, the classical mechanical analogue of the partition function is the action functional. This suggests the relation of Z-invariance to the concept of closedness of the Lagrangian 2-form, at least at the heuristic level. This relation has been made mathematically precise for a number of models, through the quasiclassical limit \cite{BMS1, BMS2}.

\item The notion of \emph{variational symmetry}, going back to the seminal work of E.~Noether \cite{N}. One only rarely finds an adequate account of her result in the modern literature. Even the classical textbooks like the Arnold's one \cite{A} only present very particular and restricted versions of Noether's theorem which deal with point symmetries. Relevant for us is the most general (original) version, dealing with what is nowadays called generalized symmetries. One of the best sources treating Noether's theorem in its full generality is Olver's textbook \cite{Olver}, which we closely follow in our presentation. In the context of integrable partial differential equations, a direct relation between Noether's theorem and closedness of the Lagrangian form in the multi-time has been shown in \cite{S2}. Here, we elaborate on this topic further, in the context of classical mechanics.
\end{itemize}

The structure of the paper is as follows. In Section  \ref{Sect var sym}, we recall, following mainly \cite[Chapter 5]{Olver}, the notions of generalized vector fields and of variational symmetries, and give a proof of Noether's theorem in the context of classical mechanical systems. In Section \ref{Sect Ham}, we provide the reader with the Hamiltonian counterpart of the theory, and show that variational symmetries of a Lagrangian system are in a one-to-one correspondence with integrals of the corresponding Hamiltonian system. Surprizingly, this fundamental result, well-known in the folklore, is difficult to locate in the literature. Our presentation is pretty close to \cite{PV}. In Section \ref{Sect comm var sym}, we discuss commuting variational symmetries, which are Lagrangian counterparts of commuting integrals of a given Hamiltonian system. Again, even if our results on commuting variational symmetries will be of no surprise for an expert, we were unable to locate a reference which would contain a satisfying presentation, and therefore we include complete details here. In Section \ref{sect from to}, we undertake a slight but a fundamentally important change of the viewpoint on variational symmetries, from an algebraic interpretation as derivations acting on differential polynomials to a geometric interpretation as commuting flows.  Of course, this latter interpretation is not less classical, but, amazingly, it seems to be usually suppressed in the literature on symmetries of differential equations. This geometric interpretation is the crucial link to the emerging theory of pluri-Lagrangian systems. In Section \ref{Sect pluri}, we recall, following \cite{S2}, the main positions of this theory. In Section \ref{Sect pluri config space}, we establish the pluri-Lagrangian structure for a general Lagrangian system possessing commuting variational symmetries. Section \ref{Sect almost closed} is devoted to some technical computations whose results illuminate the main property of the pluri-Lagrangian structure, namely the almost closedness of the pluri-Lagrangian 1-form on the space of solutions.  In Section \ref{Sect phase space}, we demonstrate that considering the pluri-Lagrangian structure in the phase space (rather than in the configuration space) allows us to greatly simplify all the concepts and constructions. The final Section \ref{Sec examples} illustrates the concepts and results of the paper with two classical examples of integrable systems, the Kepler system and the Toda lattice.

\section{Variational symmetries and Noether theorem}
\label{Sect var sym}

Let $x=(x_1,\dots,x_N)$  be coordinates on an $N$-dimensional configuration space $\mathcal{X}$. We mainly consider the case $\mathcal{X}=\R^{N}$, but $\mathcal{X}$ can also  be a smooth $N$-dimensional manifold with local coordinates $x$. Let $(x,\dot x)$ be the natural coordinates on the $2N$-dimensional tangent bundle $T\mathcal{X}$, and let  $L: T \mathcal{X} \to \R$ be a Lagrange function. The Hamilton's principle states that motions of the corresponding mechanical system are critical curves of the action functional $S$, which assigns to each curve $x:[t_0,t_1]\to\mathcal X$ (with fixed values of $x(t_0)$ and $x(t_1)$) the number
\beq  \label{lp}
S[x]=\int_{t_0}^{t_1} L(x(t), \dot x(t)) \dd t.
\eeq
Critical curves of the action (\ref{lp})  are solutions of the Euler-Lagrange equations given by
\beq\label{eq: EL 1}
\EE_i=\EE_i(x,\dot x,\ddot x):= \frac{\pa L(x, \dot x)}{\pa x_i}- \D_{t}\!
\left(\frac{\pa L(x, \dot x)}{\pa \dot x_i}\right)=0, \qquad i=1,\dots,N. 
\eeq
Here and everywhere below $\D_{t}$ denotes the total derivative w.r.t.~time $t$.
We will assume that $L$ is non-degenerate:
\beq \label{eq: nondeg}
\det\left(\frac{\partial^2 L}{\partial \dot x_i\partial \dot x_j}\right)\neq 0,
\eeq
so that the Euler-Lagrange equations can be solved for $\ddot x_i$, resulting in explicit second order differential equations.

Consider an evolutionary generalized vector field $v_1$ on $\mathcal{X}$,
\begin{equation} \label{eq: gen vector field}
v_1= \sum_{i=1}^N V_i ^{(1)}(x, \dot x)\frac{\pa}{\pa x_i}\ ,
\end{equation}
where the $N$-tuple of functions $V_i^{(1)} (x, \dot x)$, $i=1,\dots,N$, is called the characteristic  of $v_1$ (the index ``1'' will become important later, when we consider several generalized vector fields simultaneously). It generates a linear differential operator acting on differential polynomials, i.e., on functions depending on $x$ and its time derivatives $\dot x$, $\ddot x$, etc., defined by the formal sum
\beq
\D_{v_1}=\sum_{i=1}^N V_i^{(1)} \frac{\pa}{\pa x_i}+ \sum_{i=1}^N \Big(\D_{t}  V_i^{(1)}\Big) \frac{\pa}{\pa \dot x_i}+
\sum_{i=1}^N \Big(\D_{t}^2  V_i^{(1)}\Big) \frac{\pa}{\pa \ddot x_i}+\dots
\label{dgwrhg}
\eeq
Observe that $\D_{t} V_i^{(1)}$ is a function of $x$, $\dot x$ and $\ddot x$, $\D_{t}^2 V_i^{(1)}$ is a function of $x$, $\dot x$, $\ddot x$ and $\dddot x$, and so on. 

\begin{definition} {\bf (Variational symmetry)}
We say that $v_1$ is a {\rm{variational symmetry}} of the Lagrangian problem (\ref{lp}) if there exists a function $F_1=F_1(x, \dot x)$, called the 
{\rm{flux}} of the variational symmetry $v_1$, such that
\beq  \label{eq: var sym}
\D_{v_1} L (x, \dot x) - \D_{t} F_1(x, \dot x)=0.
\eeq
\end{definition}

\begin{definition} {\bf (Integral and its characteristic)}
We say that $J_1(x,\dot x)$ is an {\rm{integral of the Euler-Lagrange equations}} $\EE_i=0$ (see (\ref{eq: EL 1})) with the {\rm{characteristic}} $\big(V^{(1)}_i(x,\dot x)\big)_{1\le i\le N}$, if 
\beq  \label{eq: Noether integral char}
\D_{t} J_1(x, \dot x)=-\sum_{i=1}^N V_i^{(1)} (x, \dot x) \EE_i\ .
\eeq
\end{definition}

The famous theorem by Emmy Noether establishes a one-to-one correspondence between variational symmetries and integrals of Lagrangian systems. 

\begin{theorem} \label{Th: Noether} {\bf (E. Noether's theorem)}
\quad 

{\rm{a)}} Let the generalized vector field $v_1$ given by equation (\ref{eq: gen vector field}) be a variational  symmetry of the Lagrangian problem (\ref{lp}), with the flux $F_1(x,\dot x)$.  Then the function
\beq \label{eq: Noether integral}
J_1(x, \dot x)= \sum_{i=1}^N \frac{\partial L(x, \dot x)}{\partial \dot x_i}  V_i^{(1)}(x, \dot x)-F_1(x, \dot x)
\eeq
is an integral of Euler-Lagrange equations $\EE_i=0$ with the characteristic $\big(V^{(1)}_i(x,\dot x)\big)_{1\le i\le N}$.

\smallskip

{\rm{b)}} Conversely, let $J_1(x,\dot x)$ be an integral of motion of Euler-Lagrange equations $\EE_i=0$, with the characteristic $\big(V^{(1)}_i(x,\dot x)\big)_{1\le i\le N}$. Then the generalized vector field $v_1$ given by equation (\ref{eq: gen vector field})  is a variational symmetry of the Lagrangian problem (\ref{lp}), with the flux
\beq\label{eq: flux}
F_1(x,\dot x)= \sum_{i=1}^N \frac{\partial L(x, \dot x)}{\partial \dot x_i}  V_i^{(1)}(x, \dot x)-J_1(x, \dot x).
\eeq
\end{theorem}

\benpf
Both parts are consequences  of the following computation:
\bea
\D_{v_1} L  &=& 
\sum_{i=1}^N  V_i ^{(1)}\frac{\pa L}{\pa x_i}+ \sum_{i=1}^N \Big(\D_{t} V_i^{(1)}\Big) \frac{\pa L}{\pa \dot x_i}\nn  \\
&=& \sum_{i=1}^N  V_i ^{(1)} \!\left(\EE_i+\D_t \! \left(\frac{\pa L}{\pa \dot x_i}\right)\right) + \sum_{i=1}^N \Big(\D_{t} V_i^{(1)}\Big) \frac{\pa L}{\pa \dot x_i}\nn\\
&=& \D_t\left(\sum_{i=1}^N \frac{\partial L}{\partial \dot x_i} V_i^{(1)}\!\right)+\sum_{i=1}^N V_i^{(1)} \EE_i \ .  \nn 
\eea
The theorem is proved.
\endpf 

{\bf Remark.} {\bf (Energy integral)} \ The generalized vector field
\[
v=\sum_{i=1}^N \dot x_i \frac{\partial}{\partial x_i},
\]
which generates the differential operator $\D_v=\D_t$ (when acting on functions which do not explicitly depend on $t$), is trivially a variational symmetry of any Lagrange function $L(x,\dot x)$, with the flux $F(x,\dot x)=L(x,\dot x)$. For this variational symmetry, the Noether integral turns into the {\em energy integral}
\beq \label{eq: energy}
J(x,\dot x)=\sum_{i=1}^N \frac{\partial L}{\partial \dot x_i} \dot x_i-L(x,\dot x).
\eeq

\section{Hamiltonian side of the picture}
\label{Sect Ham}

Introduce the conjugate momenta $p=(p_1,\dots,p_N)$ by the Legendre transformation $T\mathcal X\ni (x,\dot x)\mapsto (x,p)\in T^*\mathcal X$,
\beq\label{eq: Legendre}
p_i =p_i(x,\dot x) = \frac{\pa L(x, \dot x)}{\pa \dot x_i}, \qquad i=1,\dots,N.
\eeq
The cotangent bundle $T^*\mathcal X$ is a $2N$-dimensional symplectic manifold equipped with the canonical symplectic 2-form
$$
\omega= \sum_{i=1}^N \dd x_i \wedge \dd p_i.
$$
Under the non-degeneracy condition (\ref{eq: nondeg}), equation (\ref{eq: Legendre}) defines a local diffeomorphism $T\mathcal X\to T^*\mathcal X$, which can be locally inverted to express the velocities $\dot x_i$ in terms of $(x,p)$, i.e., $\dot x_i=\dot x_i(x,p)$. The dynamics on $T^*\mathcal{X}$, equivalent to that governed by Euler-Lagrange equations (\ref{eq: EL 1}) on $T\mathcal X$, is described in terms of Hamiltonian equations of motion:
\beq  \label{eq: Ham eq}
\D_t x_i= \frac{\pa H(x,p)}{\pa p_i}, \qquad
\D_{t} p_i= -\frac{\pa H(x,p)}{\pa x_i}, \qquad i=1,\dots,N.
\eeq
Here $H: T^*\mathcal{X}\to \R$ is the Hamilton function corresponding to the Lagrange function $L$ by means of a Legendre transformation:
\beq   \label{eq: Ham}
H(x,p)= \sum_{i=1}^N p_i \dot x_i - L(x, \dot x) \Big |_{\dot x=\dot x(x,p)}.
\eeq
Derivation of Hamiltonian equations of motion (\ref{eq: Ham eq}) is based on the following facts:
\beq \label{eq: L to H}
\frac{\pa H(x,p)}{\pa x_i} = - \frac{\pa L(x,\dot x)}{\pa x_i}, \qquad \frac{\pa H(x,p)}{\pa p_i} = \dot x_i\ ,
\eeq
which hold true at the points $(x,p)\in T^*\mathcal X$ and $(x,\dot x)\in T\mathcal X$ related by the Legendre transformation (\ref{eq: Legendre}) and which follow from the definition (\ref{eq: Ham}) by differentiation.
The Hamilton function $H(x,p)$ is nothing but the energy integral $J(x,\dot x)$ expressed in terms of $(x,p)$.  The fact that $H(x,p)$ is an integral of motion of the Hamiltonian equations (\ref{eq: Ham eq}) is easily checked by direct computation. 

The following two theorems, which are in a sense mutually converse, show that variational symmetries of a Lagrangian system are in a one-to-one correspondence with integrals of the corresponding Hamiltonian system.

\begin{theorem} \label{Th: Ham converse Noether} {\bf (From a variational symmetry to a commuting Hamiltonian flow)} 
Let $v_1$ be a variational symmetry of the Lagrange function $L(x,\dot x)$, with  the flux $F_1(x,\dot x)$ and with the Noether integral $J_1(x,\dot x)$. Set $H_1(x,p)=J_1(x,\dot x)|_{\dot x=\dot x(x,p)}$. Then
\beq \label{eq: Poisson bracket}
\{ H,H_1\}=\sum_{i=1}^N \left( \frac{\pa H}{\pa x_i}\frac{\pa H_1}{\pa p_i}- \frac{\pa H}{\pa p_i}\frac{\pa H_1}{\pa x_i}\right)=0,
\eeq
so that $H_1$ is an integral of motion of the Hamiltonian flow with the Hamilton function $H$. 
\end{theorem}
\benpf
First of all, we prove that
\begin{align}
\label{eq: LQ1 to H1 x}
    \frac{\pa H_1(x,p)}{\pa x_i} & = \sum_{j=1}^N p_j\frac{\partial V_j^{(1)}(x,\dot x)}{\partial x_i}- \frac{\pa F_1(x,\dot x)}{\pa x_i}, \\
\label{eq: LQ1 to H1 p}    
    \frac{\pa H_1(x,p)}{\pa p_i} & = V_i^{(1)}(x,\dot x).
\end{align}
Indeed, we compute:
\begin{eqnarray*}
\frac{\pa H_1(x,p)}{\pa x_i} & =  & \sum_{j=1}^N p_j\frac{\partial V_j^{(1)}}{\partial x_i}
                                                    +\sum_{j=1}^N\sum_{k=1}^N p_j\frac{\partial V_j^{(1)}}{\partial \dot x_k}  \frac{\partial \dot x_k}{\partial x_i}
                                                    -\frac{\partial F_1}{\partial x_i}-\sum_{k=1}^N \frac{\partial F_1}{\partial \dot x_k} \frac{\partial \dot x_k}{\partial x_i}\ , \\
\frac{\pa H_1(x,p)}{\pa p_i} & = & V_i^{(1)}
                                                    +\sum_{j=1}^N\sum_{k=1}^N p_j\frac{\partial V_j^{(1)}}{\partial \dot x_k}  \frac{\partial \dot x_k}{\partial p_i}
                                                    -\sum_{k=1}^N \frac{\partial F_1}{\partial \dot x_k} \frac{\partial \dot x_k}{\partial p_i}\ ,
\end{eqnarray*}
and (\ref{eq: LQ1 to H1 x}), (\ref{eq: LQ1 to H1 p}) follow by virtue of the following Lemma:
\begin{lemma} \label{lemma F}
For the flux $F_1(x,\dot x)$ of a variational symmetry (\ref{eq: gen vector field}), one has:
\beq  \label{eq: dFdxdot}
\frac{\partial F_1}{\partial \dot x_i}=\sum_{j=1}^N \frac{\partial L}{\partial \dot x_j}\ \frac{\partial V_j^{(1)}}{\partial \dot x_i}.
\eeq
\end{lemma}
\noindent
{\bf Proof of Lemma \ref{lemma F}.} We have:
\begin{eqnarray*}
\D_tF_1 & = & \sum_{i=1}^N \frac{\partial F_1}{\partial x_i} \ \dot x_i+\sum_{i=1}^N \frac{\partial F_1}{\partial \dot x_i} \ \ddot x_i\ , \\
\D_{v_1}L & = & \sum_{j=1}^N \frac{\partial L}{\partial x_j} V_j^{(1)}+\sum_{j=1}^N\frac{\partial L}{\partial \dot x_j} \big(\D_t V_j^{(1)}\big)\\
         & = & \sum_{j=1}^N \frac{\partial L}{\partial x_j} V_j^{(1)}
         +\sum_{j=1}^N\sum_{i=1}^N\frac{\partial L}{\partial \dot x_j}\ \frac{\partial V_j^{(1)}}{\partial x_i} \ \dot x_i
         +\sum_{j=1}^N\sum_{i=1}^N\frac{\partial L}{\partial \dot x_j} \ \frac{\partial V_j^{(1)}}{\partial \dot x_i} \ \ddot x_i. 
\end{eqnarray*}
Equation \eqref{eq: dFdxdot} follows by comparison of coefficients by $\ddot x_i$. \qed
\medskip

Continuing the proof of Theorem \ref{Th: Ham converse Noether}, we compute with the help of (\ref{eq: L to H}), (\ref{eq: LQ1 to H1 x}) and (\ref{eq: LQ1 to H1 p}):
\begin{eqnarray*}
\{H_1,H\} & = & \sum_{i=1}^N \left( -\frac{\pa H}{\pa x_i}\frac{\pa H_1}{\pa p_i}+ \frac{\pa H}{\pa p_i}\frac{\pa H_1}{\pa x_i}\right) \\
 & = & \sum_{i=1}^N\left(\frac{\partial L}{\partial x_i} V_i^{(1)}+\dot x_i\Big(\sum_{j=1}^N p_j\frac{\partial V_j^{(1)}}{\partial x_i}- \frac{\pa F_1}{\pa x_i}\Big)\right)\\
 & = & \sum_{i=1}^N (\EE_i+\D_tp_i)V_i^{(1)}+\sum_{j=1}^N p_j\big(\D_tV_j^{(1)}\big)-\D_tF_1-
 \sum_{i=1}^N \ddot x_i\bigg(\sum_{j=1}^N p_j\frac{\partial V_j^{(1)}}{\partial \dot x_i}- \frac{\pa F_1}{\pa \dot x_i}\bigg).
 \end{eqnarray*}
 According to Lemma \ref{lemma F}, the last sum vanishes, and we find:
 \[
 \{H_1,H\}=\sum_{i=1}^N V_i^{(1)}\EE_i+\D_t\Big(\sum_{i=1}^N p_iV_i^{(1)}-F_1\Big)=\sum_{i=1}^N V_i^{(1)}\EE_i+\D_tJ_1=0.
 \]
Theorem \ref{Th: Ham converse Noether} is proved.
\endpf

{\bf Remark.} Note that equations (\ref{eq: L to H}) can be recovered from equations (\ref{eq: LQ1 to H1 x}), (\ref{eq: LQ1 to H1 p}), if we replace $V_i^{(1)}(x,\dot x)$ and $F_1(x,\dot x)$ by $V_i(x,\dot x)=\dot x_i$ and $F(x,\dot x)=L(x,\dot x)$, respectively.

\begin{theorem} {\bf (From a commuting Hamiltonian flow to a variational symmetry)}  
Let $H_1: T^*\mathcal{X}\to \R$ be an integral of motion of the Hamiltonian flow with the Hamilton function $H$. 
Set
$$
V_i^{(1)}(x, \dot x)= \left.\frac{\pa H_1(x,p)}{ \pa p_i} \right|_{p=p(x,\dot x)}\ ,
$$
and define a generalized vector field $v_1=\sum_{i=1}^N V_i^{(1)}\partial/\partial x_i$. Then $v_1$ is a variational symmetry of the Lagrangian problem (\ref{lp}), with the Noether integral $J_1(x,\dot x)$ given by
$$
J_1(x,\dot x)=H_1(x,p) \Big |_{p=p(x,\dot x)},
$$
and with the flux $F_1(x,\dot x)$ given by formula (\ref{eq: flux}). 
\end{theorem}

\benpf
According to Theorem \ref{Th: Noether}, part b), it is enough to show that $J_1$ is an integral of motion of equations $\EE_i=0$ with the characteristic $\big(V_i^{(1)}\big)_{1\le i \le N}$, so that equation (\ref{eq: Noether integral char}) is satisfied. But this follows by a direct computation (in the following formulas we always assume that $(x,p)\in T^*\mathcal X$ and $(x,\dot x)\in T\mathcal X$ are related by the Legendre transformation (\ref{eq: Legendre})):
\begin{eqnarray*}
\D_{t} J_1  &=&  \sum_{i=1}^N \left(\frac{\pa H_1(x,p)}{\pa x_i} \dot x_i + \frac{\pa H_1(x,p)}{\pa p_i} \D_{t} p_i(x, \dot x) \right) \\
                  &=&  \sum_{i=1}^N \left(\frac{\pa H_1(x,p)}{\pa x_i}\dot x_i + \frac{\pa H_1(x,p)}{\pa p_i} \left(\frac{\pa L(x,\dot x)}{\pa x_i}-\EE_i\right) \right) \\
                 &=&  \sum_{i=1}^N \left(\frac{\pa H_1}{\pa x_i}\frac{\pa H}{\pa p_i}-\frac{\pa H_1}{\pa p_i}\frac{\pa H}{\pa x_i} \right)
                 -\sum_{i=1}^N\frac{\pa H_1(x,p)}{ \pa p_i}  \EE_i  \\
                & = & \{H_1,H\}-\sum_{i=1}^N V_i^{(1)}\EE_i\ ,
\end{eqnarray*}
where we used (\ref{eq: EL 1}) and (\ref{eq: L to H}). It remains to use (\ref{eq: Poisson bracket}). \endpf

\section{Commuting variational symmetries}
\label{Sect comm var sym}

We show that commutativity of variational symmetries provides us with an adequate framework to discuss integrability in the Lagrangian context. Recall that  usually the notion of integrability is expressed using the Hamiltonian language.

Assume that the Lagrange function $L(x,\dot x)$ admits two variational symmetries
\[
v_k=\sum_{i=1}^N V_i^{(k)}(x,\dot x)\frac{\partial}{\partial x_i}\ , \qquad k=1,2,
\]
with the corresponding fluxes $F_k(x,\dot x)$, so that
\begin{equation} \label{eq: var symmetries}
\D_{v_k} L(x,\dot x)-\D_t F_k(x,\dot x)=0, \qquad k=1,2.
\end{equation}
Recall \cite{Olver} that the commutator $[\D_{v_1},\D_{v_2}]$ is a differential operator which can be represented as $\D_{[v_1,v_2]}$, that is, the prolongation of the generalized vector field 
\beq \label{eq: commutator}
[v_1,v_2]=\sum_{i=1}^N \Big(\D_{v_1} V_i^{(2)}-\D_{v_2} V_i^{(1)}\Big)\frac{\partial}{\partial x^i}\ .
\eeq

\begin{definition}{\bf (Commuting variational symmetries)}  
We say that variational symmetries $v_1$, $v_2$ of the variational problem (\ref{lp}) {\rm{commute}} if 
\beq \label{eq: commuting sym}
\D_{v_1} V_i^{(2)}-\D_{v_2} V_i^{(1)} = \sum_{j=1}^N  r_{ij}(x,\dot x)\EE_j\ , \quad i=1,\ldots, N
\eeq
with some functions $r_{ij}(x,\dot x)$, so that $[v_1,v_2]=0$ on solutions of the Euler-Lagrange equations (\ref{eq: EL 1}).
\end{definition}
\begin{proposition}  \label{lemma skew}
For two commuting variational symmetries $v_1$, $v_2$, the functions $r_{ij}(x,\dot x)$ are skew-symmetric,
\beq \nn 
r_{ij}(x,\dot x)=-r_{ji}(x,\dot x).
\eeq
More precisely, setting $H_k(x,p)=J_k(x,\dot x)|_{\dot x=\dot x(x,p)}$, we have:
\beq  \label{eq: rij}
r_{ij}(x,\dot x)= \sum_{k,m=1}^N\left(
\frac{\pa^2 H_1}{ \pa p_i \pa p_m} \ \frac{\pa^2 L}{\pa \dot x_m \pa \dot x_k} \ \frac{\pa^2 H_2}{ \pa p_k \pa p_j} -
\frac{\pa^2 H_2}{ \pa p_i \pa p_m} \ \frac{\pa^2 L}{\pa \dot x_m \pa \dot x_k} \ \frac{\pa^2 H_1}{ \pa p_k \pa p_j}
\right).
\eeq
\end{proposition}
\benpf
We compare the terms with $\ddot x_\ell$ on the both sides of equation (\ref{eq: commuting sym}). Such terms in $\D_{v_1} V_i^{(2)}-\D_{v_2} V_i^{(1)}$ come from 
\[
\sum_{k=1}^N\left(\big(\D_tV_k^{(1)}\big)\frac{\partial V_i^{(2)}}{\partial \dot x_k}-\big(\D_tV_k^{(2)}\big)\frac{\partial V_i^{(1)}}{\partial \dot x_k}\right),
\]
and are equal to 
\[
\sum_{k=1}^N\sum_{\ell=1}^N\left(\frac{\partial V_i^{(2)}}{\partial \dot x_k}\ \frac{\partial V_k^{(1)}}{\partial \dot x_\ell}
                                                   -\frac{\partial V_i^{(1)}}{\partial \dot x_k}\ \frac{\partial V_k^{(2)}}{\partial \dot x_\ell}\right)\! \ddot x_\ell\ .
\]
From (\ref{eq: LQ1 to H1 p}) we derive:
\[
\frac{\partial V_i^{(1)}}{\partial \dot x_k}=\sum_{m=1}^N \frac{\partial^2 H_1}{\partial p_i\partial p_m}\ \frac{\partial p_m}{\partial \dot x_k}.
\]
Thus, the terms with $\ddot x_\ell$ in $\D_{v_1} V_i^{(2)}-\D_{v_2} V_i^{(1)}$ are equal to
\[
\sum_{k=1}^N\sum_{\ell=1}^N \sum_{m=1}^N\sum_{j=1}^N\left(
   \frac{\partial^2 H_2}{\partial p_i\partial p_m}\ \frac{\partial p_m}{\partial \dot x_k}
   \ \frac{\partial^2 H_1}{\partial p_k\partial p_j}\ \frac{\partial p_j}{\partial \dot x_\ell}
    -\frac{\partial^2 H_1}{\partial p_i\partial p_m}\ \frac{\partial p_m}{\partial \dot x_k}
    \ \frac{\partial^2 H_2}{\partial p_k\partial p_j}\ \frac{\partial p_j}{\partial \dot x_\ell}\right)\! \ddot x_\ell\ .
\]
This has to be compared with the terms with $\ddot x_\ell$ in $\sum_{j=1}^N r_{ij}\EE_j$. According to (\ref{eq: EL 1}), we have:
\[
\EE_j=-\sum_{\ell=1}^N \frac{\partial p_j}{\partial \dot x_\ell}\ \ddot x_\ell+\ldots.
\]
With a reference to the non-degeneracy condition (\ref{eq: nondeg}), we arrive at (\ref{eq: rij}).
\endpf
\begin{proposition} \label{prop comm symm}
If two variational symmetries $v_1$, $v_2$ commute, then their fluxes $F_1$, $F_2$ satisfy:
\beq
\D_{v_1} F_2-\D_{v_2} F_1=c_{12}+\sum_{i=1}^Np_i\left(\D_{v_1} V_i^{(2)}-\D_{v_2} V_i^{(1)}\right),
\eeq
where $c_{12}={\rm const}$. In particular, on solutions of the Euler-Lagrange equations $\EE_j=0$ we have:
\beq \label{eq: D1F2-D2F1}
\D_{v_1} F_2-\D_{v_2} F_1=c_{12}\ .
\eeq
\end{proposition}
\benpf
Using definition  (\ref{eq: var symmetries}) of variational symmetries and the fact that evolutionary vector fields $\D_{v_k}$ commute with $\D_t$ (see \cite{Olver}), we find:
\begin{eqnarray*}
\lefteqn{\D_t\Big(\D_{v_1} F_2-\D_{v_2} F_1\Big) \ = \ [\D_{v_1},\D_{v_2}] L \ = \ \D_{[v_1,v_2]}L}\\
   & = & \sum_{i=1}^N  \bigg(\!\!\left(\D_{v_1} V_i^{(2)}-\D_{v_2} V_i^{(1)}\right)\frac{\partial L}{\partial x_i}
   +\D_t\! \left(\D_{v_1} V_i^{(2)}-\D_{v_2} V_i^{(1)}\right)\frac{\partial L}{\partial \dot x_i}\bigg)\\
    & = & \sum_{i=1}^N  \bigg(\!\! \left(\D_{v_1} V_i^{(2)}-\D_{v_2} V_i^{(1)}\right)\! (\EE_i 
    + \D_t p_i)+\D_t\! \left(\D_{v_1} V_i^{(2)}-\D_{v_2} V_i^{(1)}\right)p_i\bigg)\\
    & = & \sum_{i=1}^N  \left(\D_{v_1} V_i^{(2)}-\D_{v_2} V_i^{(1)}\right)\EE_i 
    + \D_t \bigg(\sum_{i=1}^N p_i\left(\D_{v_1} V_i^{(2)}-\D_{v_2} V_i^{(1)}\right)\!\bigg).
\end{eqnarray*}
The first sum on the right-hand side is equal to $\sum_{i=1}^N\sum_{j=1}^N r_{ij}\EE_j\EE_i$ and vanishes due to the skew-symmetry of $r_{ij}$ (see Proposition \ref{lemma skew}).
\endpf

\begin{theorem} \label{Th: Ham converse several Noether} {\bf (From commuting variational symmetries to commuting Hamiltonian flows)}
Let $v_k$ $(k=1,2)$ be commuting variational symmetries of the Lagrange function $L(x,\dot x)$ with the Noether integrals $J_k(x,\dot x)$. Set $H_k(x,p)=J_k(x,\dot x)|_{\dot x=\dot x(x,p)}$. Then
\beq \label{eq: Poisson bracket kl}
\{ H_1,H_2\}=c_{12}\ ,
\eeq
with the constant $c_{12}$ from Proposition \ref{prop comm symm}. Thus, Hamiltonian flows with the Hamilton functions $H_k$ commute.
\end{theorem}
\benpf
We compute:
\begin{eqnarray*}
\D_{v_1}F_2-\D_{v_2} F_1 & = &  
   \sum_{i=1}^N \left(V_i^{(1)}\frac{\partial F_2}{\partial x_i}-V_i^{(2)}\frac{\partial F_1}{\partial x_i}\right)
           +\sum_{i=1}^N \left(\big(\D_tV_i^{(1)}\big)\frac{\partial F_2}{\partial \dot x_i}-\big(\D_tV_i^{(2)}\big)\frac{\partial F_1}{\partial \dot x_i}\right).
\end{eqnarray*}
Using Lemma \ref{lemma F}, we find:
\begin{eqnarray*}
\D_{v_1}F_2-\D_{v_2} F_1 & = & \sum_{i=1}^N \left(V_i^{(1)}\frac{\partial F_2}{\partial x_i}-V_i^{(2)}\frac{\partial F_1}{\partial x_i}\right)\nn \\
       &&    +\sum_{i=1}^N \sum_{j=1}^Np_j\left(\big(\D_tV_i^{(1)}\big)\frac{\partial V_j^{(2)}}{\partial \dot x_i}-\big(\D_tV_i^{(2)}\big)\frac{\partial V_j^{(1)}}{\partial \dot x_i}\right)\\
      & = & \sum_{i=1}^N \left(V_i^{(1)}\frac{\partial F_2}{\partial x_i}-V_i^{(2)}\frac{\partial F_1}{\partial x_i}\right)
               -\sum_{i=1}^N \sum_{j=1}^Np_j\left(V_i^{(1)}\frac{\partial V_j^{(2)}}{\partial x_i}-V_i^{(2)}\frac{\partial V_j^{(1)}}{\partial x_i}\right) \\
      &    & +\sum_{j=1}^Np_j\Big(\D_{v_1} V_j^{(2)}-\D_{v_2} V_j^{(1)}\Big) .           
\end{eqnarray*}
Comparing with Proposition \ref{prop comm symm}, we find:
\[
 \sum_{i=1}^N \left(V_i^{(2)}\bigg(\sum_{j=1}^N p_j\frac{\partial V_j^{(1)}}{\partial x_i}-\frac{\partial F_1}{\partial x_i}\bigg)
                      -V_i^{(1)}\bigg(\sum_{j=1}^N p_j\frac{\partial V_j^{(2)}}{\partial x_i}-\frac{\partial F_2}{\partial x_i}\bigg)\right) = c_{12}\ .           
\]
According to (\ref{eq: LQ1 to H1 x}), (\ref{eq: LQ1 to H1 p}), this can be put as 
\[
\sum_{i=1}^N \left(\frac{\partial H_2}{\partial p_i} \frac{\partial H_1}{\partial x_i}-
                      \frac{\partial H_1}{\partial p_i} \frac{\partial H_2}{\partial x_i}\right) =\{H_1,H_2\}=c_{12}\ ,
\]
which finishes the proof.
\endpf

\begin{theorem} \label{Th: Ham converse several Noether} {\bf (From commuting Hamiltonian flows to commuting variational symmetries)}
Let $H_1, H_2: T^*\mathcal{X}\to \R$ be integrals of motion of the Hamiltonian flow with the Hamilton function $H$, such that the corresponding Hamiltonian flows commute, so that (\ref{eq: Poisson bracket kl}) is satisfied.
Set
$$
V_i^{(k)}(x, \dot x)= \left.\frac{\pa H_k(x,p)}{ \pa p_i} \right|_{p=p(x,\dot x)}\ ,\quad k=1,2,
$$
and define variational symmetries $v_k=\sum_{i=1}^N V_i^{(k)}\partial /\partial x_i$ $(k=1,2)$ of the Lagrangian problem (\ref{lp}). Then $v_1$, $v_2$ commute on solutions of the Lagrangian problem (\ref{lp}).
\end{theorem}

\benpf
We compute:
\begin{eqnarray*}
\D_{v_1} V_i^{(2)} 
   & = &  \sum_{k=1}^N \left( \frac{\pa V_i^{(2)}}{\pa x_k} V_k^{(1)}  + \frac{\pa V_i^{(2)}}{\pa \dot x_k} \big(\D_t V_k^{(1)}\big)\right) \\
   & = & \sum_{k=1}^N \frac{\pa V_i^{(2)}}{\pa x_k} V_k^{(1)} 
             +\sum_{k=1}^N \sum_{\ell=1}^N
             \frac{\pa V_i^{(2)}}{\pa \dot x_k} \left( \frac{\pa V_k^{(1)}}{\pa x_\ell} \dot x_\ell + \frac{\pa V_k^{(1)}}{\pa \dot x_\ell} \ddot x_\ell \right) \! .
 \end{eqnarray*}
 Differentiating the definition $V_i^{(k)}=\partial H_k/\partial p_i$, we find:
\begin{eqnarray*}
\lefteqn{ \D_{v_1} V_i^{(2)}  
    =  \sum_{k=1}^N \left( \frac{\pa^2 H_2}{\pa p_i \pa x_k} 
            +\sum_{m=1}^N  \frac{\pa^2 H_2}{ \pa p_i \pa p_m} \ \frac{\pa p_m}{\pa x_k}\right)\frac{\pa H_1}{\pa p_k} }\nn \\
   &    & +\sum_{k=1}^N \sum_{m=1}^N \frac{\pa^2 H_2}{ \pa p_i \pa p_m} \ \frac{\pa p_m}{\pa \dot x_k}
             \sum_{\ell=1}^N  \left(\frac{\pa^2 H_1}{ \pa p_k \pa x_\ell} \ \dot x_\ell 
            +\sum_{j=1}^N \frac{\pa^2 H_1}{ \pa p_k \pa p_j} \bigg( \frac{\pa p_j }{\pa x_\ell}\ \dot x_\ell +\frac{\pa p_j }{\pa \dot x_\ell} \ \ddot x_\ell \bigg)\right).
\end{eqnarray*}
In the last line we use 
$$
\sum_{\ell=1}^N \left( \frac{\pa p_j}{\pa x_\ell} \ \dot x_\ell+ \frac{\pa p_j}{\pa \dot x_\ell} \ \ddot x_\ell \right)=\D_tp_j=
\frac{\pa L}{\pa x_j}-\EE_j\ ,
$$
and obtain:
\begin{eqnarray*}
\lefteqn{ \D_{v_1} V_i^{(2)}  
    =  \sum_{k=1}^N \left( \frac{\pa^2 H_2}{\pa p_i \pa x_k} 
            +\sum_{m=1}^N  \frac{\pa^2 H_2}{ \pa p_i \pa p_m} \ \frac{\pa p_m}{\pa x_k}\right)\frac{\pa H_1}{\pa p_k} }\nn \\
   &    & +\sum_{k=1}^N \sum_{m=1}^N \frac{\pa^2 H_2}{ \pa p_i \pa p_m} \ \frac{\pa p_m}{\pa \dot x_k}
             \left( \sum_{\ell=1}^N \frac{\pa^2 H_1}{ \pa p_k \pa x_\ell} \ \dot x_\ell 
            +\sum_{j=1}^N \frac{\pa^2 H_1}{ \pa p_k \pa p_j} \ \frac{\pa L}{\pa x_j} -\sum_{j=1}^N \frac{\pa^2 H_1}{ \pa p_k \pa p_j} \EE_j\right).
\end{eqnarray*}
In the last line we replace $\dot x_\ell$ by $\pa H / \pa p_\ell$ and $\pa L/\pa x_j$ by $-\pa H / \pa x_j$:
\begin{eqnarray*}
\D_{v_1} V_i^{(2)}
  & = & \sum_{k=1}^N \frac{\pa^2 H_2}{\pa p_i \pa x^k}\ \frac{\pa H_1}{\pa p_k}
          +\sum_{k=1}^N \sum_{m=1}^N \frac{\pa^2 H_2}{ \pa p_i \pa p_m} \ \frac{\pa^2 L}{\pa \dot x_m\pa x_k} \ \frac{\pa H_1}{\pa p_k}  \\
  &    & +\sum_{k=1}^N \sum_{m=1}^N \frac{\pa^2 H_2}{ \pa p_i \pa p_m} \  \frac{\pa^2 L}{\pa \dot x_m\pa \dot x_k}
  \sum_{j=1}^N \left(\frac{\pa^2 H_1}{  \pa p_k\pa x_j} \  \frac{\pa H}{\pa p_j}  - \frac{\pa^2 H_1}{ \pa p_k \pa p_j} \ \frac{\pa H}{\pa x_j} \right)  \\
  &    & -\sum_{k=1}^N \sum_{m=1}^N \sum_{j=1}^N \frac{\pa^2 H_2}{ \pa p_i \pa p_m} \ \frac{\pa^2 L}{\pa \dot x_m\pa \dot x_k} \  
            \frac{\pa^2 H_1}{ \pa p_k \pa p_j}  \ \EE_j\ .
\end{eqnarray*}
In the second line of the above formula we use the identity
\[
\sum_{j=1}^N \left(\frac{\pa^2 H_1}{  \pa p_k\pa x_j} \  \frac{\pa H}{\pa p_j}  - \frac{\pa^2 H_1}{ \pa p_k \pa p_j} \ \frac{\pa H}{\pa x_j} \right) =
\sum_{j=1}^N \left(\frac{\pa H_1}{\pa p_j} \  \frac{\pa^2 H}{\pa p_k \pa x_j}  - \frac{\pa H_1}{ \pa x_j} \ \frac{\pa^2 H}{\pa p_k\pa p_j} \right),
\]
which is obtained from
\[
 \{H_1,H\} = \sum_{j=1}^N \left( \frac{\pa H_1}{\pa x_j} \ \frac{\pa H}{\pa p_j}- \frac{\pa H_1}{\pa p_j}\  \frac{\pa H}{\pa x_j}\right)=0
\]
by differentiation with respect to $p_k$. Thus, we find:
\begin{eqnarray*}
\D_{v_1} V_i^{(2)}
  & = & \sum_{k=1}^N \frac{\pa^2 H_2}{\pa p_i \pa x_k}\ \frac{\pa H_1}{\pa p_k}
          +\sum_{k=1}^N \sum_{m=1}^N \frac{\pa^2 H_2}{ \pa p_i \pa p_m} \ \frac{\pa^2 L}{\pa \dot x_m\pa x_k} \cdot \frac{\pa H_1}{\pa p_k}  \\
  &    & +\sum_{k=1}^N \sum_{m=1}^N \frac{\pa^2 H_2}{ \pa p_i \pa p_m} \  \frac{\pa^2 L}{\pa \dot x_m\pa \dot x_k}
  \sum_{j=1}^N \left(\frac{\pa H_1}{\pa p_j} \  \frac{\pa^2 H}{\pa p_k \pa x_j}  - \frac{\pa H_1}{ \pa x_j} \ \frac{\pa^2 H}{\pa p_k\pa p_j} \right),  \\
  &    & -\sum_{k=1}^N \sum_{m=1}^N \sum_{j=1}^N \frac{\pa^2 H_2}{ \pa p_i \pa p_m} \ \frac{\pa^2 L}{\pa \dot x_m\pa \dot x_k} \  
            \frac{\pa^2 H_1}{ \pa p_k \pa p_j}  \ \EE_j\ .
\end{eqnarray*}
In the second line of the above formula we use the following identities:
\beq \label{i2}
\sum_{k=1}^N \frac{\pa^2 L}{\pa \dot x_m \pa \dot x_k} \ \frac{\pa^2 H}{ \pa p_k \pa x_j}=-\frac{\pa^2 L}{\pa \dot x_m \pa x_j},
\eeq
and
\beq \label{i1}
\sum_{k=1}^N  \frac{\pa^2 L}{\pa \dot x_m \pa \dot x_k} \ \frac{\pa^2 H}{\pa p_k \pa p_j}=\delta_{mj}.
\eeq
Both are easily derived by differentiating the identity 
$$
\frac{\pa L}{\pa \dot x_m}\Big(x,\frac{\pa H(x,p)}{\pa p}\Big)=p_m
$$
with respect to $x_j$ and to $p_j$. Using (\ref{i2}) and (\ref{i1}), we finally obtain:
\begin{eqnarray} \label{55}
\D_{v_1} V_i^{(2)} & = & \sum_{k=1}^N \frac{\pa^2 H_2}{\pa p_i \pa x_k}  \ \frac{\pa H_1}{\pa p_k}
-\sum_{k=1}^N \frac{\pa^2 H_2}{ \pa p_i \pa p_k} \ \frac{\pa H_1}{\pa x_k}\nn \\
  & & -\sum_{k=1}^N \sum_{m=1}^N \sum_{j=1}^N \frac{\pa^2 H_2}{ \pa p_i \pa p_m} \ \frac{\pa^2 L}{\pa \dot x_m\pa \dot x_k} \  
            \frac{\pa^2 H_1}{ \pa p_k \pa p_j}  \ \EE_j\ .
\end{eqnarray}
By interchanging the roles of indices 1 and 2, and by subtracting the resulting formula from the previous one, we come to the final result:
\[
 \D_{v_1} V_i^{(2)}-\D_{v_2} V_i^{(1)}=\frac{\pa}{\pa p_i}  \{H_1,H_2\}+\sum_{j=1}^N r_{ij}\EE_j\ ,
\]
with $r_{ij}$ as given in (\ref{eq: rij}). This proves the theorem.
\endpf

\section{From variational symmetries to a multi-time 1-form}
\label{sect from to}

Now we would like to promote an alternative point of view. In the preceding presentation, the property of a generalized vector field $v_1$ to be a variational symmetry was mainly understood as an algebraic property of an operator $\D_{v_1}$ acting on differential polynomials (functions of $x$, $\dot x$, $\ddot x$ etc.), where $x$ was interpreted as a function of a single time $t$. However, this interpretation does not incorporate the view of a symmetry as a commuting flow. This is the point of view we would like to adopt now. 

Suppose that a Lagrangian problem (\ref{lp}) admits $m$ pairwise commuting variational symmetries $v_1,\ldots,v_m$. We interpret them as $m$ flows 
\begin{equation} \label{eq: new flows}
(x_i)_{t_k}= V_i^{(k)} (x, \dot x), \qquad i=1,\ldots,N, \quad k=1,\ldots,m.
\end{equation}
These flows commute when restricted to solutions of the variational problem (\ref{lp}) and therefore govern the evolution of the fields $x$ which can be now interpreted as functions of $m+1$ independent variables, $x=x(t,t_1,\ldots,t_m)$. Moreover, we consider derivation operators acting on differential polynomials which are now functions of $x$ and mixed derivatives of $x$ of all orders with respect to $t$ and to all $t_1,\ldots,t_m$ (derivatives with respect to $t$ being still denoted by dots). We will use the standard multi-index notation for such derivatives, with multi-indices $I=(i_0,i_1,\ldots,i_m)\in (\mathbb Z_{\ge 0})^{m+1}$. The operators of full derivatives $\D_{t}$ and $\D_{t_k}$ are defined as usual:
\begin{eqnarray*}
\D_{t} & = & \sum_{i=1}^N \sum_{I} (x_i)_{I+ e_0} \ \frac{\pa}{\pa (x_i)_I}=\sum_{i=1}^N \left( \dot x_i \ \frac{\pa}{\pa x_i}+\ddot x_i \ \frac{\pa}{\pa \dot x_i}
    +\sum_{\ell=1}^m (\dot x_i)_{t_\ell} \ \frac{\pa}{\pa (x_i)_{t_\ell}}+\ldots \right),\\
\D_{t_k} & = &  \sum_{i=1}^N \sum_{I} (x_i)_{I+ e_k} \ \frac{\pa}{\pa (x_i)_I}=\sum_{i=1}^N \left( (x_i)_{t_k} \ \frac{\pa}{\pa x_i}+(\dot x_i)_{t_k} \ \frac{\pa}{\pa \dot x_i}
    +\sum_{\ell=1}^m (x_i)_{t_\ell t_k} \ \frac{\pa}{\pa (x_i)_{t_\ell}}+\ldots \right).\qquad
\end{eqnarray*}
Operators $\D_{t_k}$ come to replace $\D_{v_k}$. In particular, when acting on a function depending on $x,\dot x$ only, like $L(x,\dot x)$ or $F_\ell(x,\dot x)$, we have:
\beq
\D_{t_k}-\D_{v_k}=\sum_{i=1}^N\big((x_i)_{t_1}-V_i^{(k)}(x,\dot x)\big)\frac{\partial}{\partial x_i}+
\sum_{i=1}^N\D_t\big((x_i)_{t_1}-V_i^{(k)}(x,\dot x)\big)\frac{\partial}{\partial \dot x_i}\ ,
\eeq
so that results coincide on solutions of the flows (\ref{eq: new flows}).

Upon the replacement of $\D_{v_k}$ by $\D_{t_k}$, the defining formula (\ref{eq: var symmetries}) of a variational symmetries would read 
$$
\D_{t_k} L(x,\dot x)-\D_{t} F_k(x,\dot x)=0,
$$
while the formulas similar to (\ref{eq: D1F2-D2F1}) would read
$$
\D_{t_k} F_\ell(x,\dot x)-\D_{t_\ell} F_k(x,\dot x)=c_{k\ell}.
$$
An important insight is that the latter relations are formally nothing but the coefficients of $\dd \mathcal L$ for the 1-form
\beq \label{eq: wrong form}
\mathcal L=L(x,\dot x)\dd t + F_1(x,\dot x)\dd t_1+\ldots+F_m(x,\dot x)\dd t_m.
\eeq
Of course, all these relations are satisfied not as algebraic identities, but only on solutions of a consistent system consisting of differential equations (\ref{eq: EL 1}) and (\ref{eq: new flows}).

\section{One-dimensional pluri-Lagrangian systems}
\label{Sect pluri}

In \cite{S1}, a theory of the so called {\em pluri-Lagrangian problems} which consist in finding critical points for action functionals associated with such forms, has been developed. Here are the main positions of this theory. In doing this, we do not single out the time $t$ playing a special role. Thus, in the present section we consider functions on the $m$-dimensional time $(t_1,\ldots,t_m)\in\mathbb R^m$.

\begin{definition} {\bf (Pluri-Lagrangian problem)} Consider a $1$-form on $\R^m$ given by
\beq \label{sgfg}
\mathcal{L}=\mathcal{L}[x, x_{t_1}, \dots, x_{t_m}]=\sum_{k=1}^m L_k (x, x_{t_1},\ldots,x_{t_m}) \dd t_k\ ,
\eeq
depending on (the first jet of) a function $x:\R^m\to \mathcal{X}$.
For any smooth curve $\Gamma:[0,1]\to \R^m$,  define the action functional
\beq
S_\Gamma= \int_\Gamma\mathcal{L}.
\eeq
A {\rm{pluri-Lagrangian problem}} for the 1-form $\mathcal L$ consists in finding functions $x:\R^m\to \mathcal{X}$ which deliver critical point of $S_\Gamma$ 
for any $\Gamma$ (for fixed values of $x$ at the endpoints of $\Gamma$).
\end{definition}

\begin{theorem} \label{Th multi-time EL eqs} {\bf (Multi-time Euler-Lagrange equations)}
A function $x:\R\to\mathcal X$ solves the pluri-Lagrangian problem for the $1$-form $\mathcal L$, if and only if it satisfies the following differential equations, called the {\em multi-time Euler-Lagrange equations}:

\beq \label{eq: pluri 1}
\frac{\pa L_k}{ \pa (x_i)_{t_\ell}}=0, \qquad \forall i=1,\dots, N, \quad  \forall k,\ell\in\{1,\dots, m\} \;\;{\text with} \;\; k\neq \ell;
\eeq

\beq\label{eq: pluri 2}
\frac{\partial L_1}{\partial (x_i)_{t_1}}=\dots=
\frac{\partial L_m}{\partial (x_i)_{t_m}}, \qquad \forall i=1,\dots, N;
\eeq
if the common value of these functions is denoted by $p_i:\R^m\to\R$, then
\beq \label{eq: pluri 3}
\D_{t_k} p_i = \frac{\partial L_k}{\partial x_i} \qquad \forall i=1,\ldots, N, \quad \forall\, k=1,\dots, m.
\eeq
\end{theorem}

 \begin{theorem} \label{Th almost closed} {\bf (Almost closedness of the pluri-Lagrangian form)}
On solutions of multi-time Euler-Lagrange equations (\ref{eq: pluri 1})-(\ref{eq: pluri 3}), we have:
$$
\D_{t_k}L_\ell-\D_{t_\ell}L_k=c_{k\ell}={\rm const}.
$$
In particular, if all $c_{k\ell}=0$,  then the 1-form (\ref{sgfg})  is closed on solutions of multi-time Euler-Lagrange equations,
so that the action functional $S_\Gamma$ does not depend on the choice of the curve
$\Gamma$ connecting two given points in $\R^m$.
\end{theorem}

\section{From variational symmetries to a pluri-Lagrangian problem}
\label{Sect pluri config space}

It is easy to check that for the 1-form (\ref{eq: wrong form}) with $m+1$ times $t,t_1,\ldots, t_m$ the pluri-Lagrangian problem is inconsistent (for instance, equation (\ref{eq: pluri 2}) is not satisfied, since $\partial F_k/\partial (x_i)_{t_k} =0\neq \partial L/\partial \dot x_i$). However, it turns out to be possible to modify this 1-form to 
\beq \label{eq: right form}
\mathcal L=L(x,\dot x)\dd t+L_1(x,\dot x,x_{t_1}) \dd t_1+\ldots +L_m(x,\dot x,x_{t_m}) \dd t_m,
\eeq
for which the pluri-Lagrangian problem is  well-posed, with multi-time Euler-Lagrange equations being exactly the desired ones. Moreover, on solutions of these equations, forms (\ref{eq: wrong form}) and (\ref{eq: right form}) coincide. To achieve this, the functions $L_k(x,\dot x,x_{t_k})$ are obtained by replacing in formula (\ref{eq: flux}) for the flux $F_k$ the quantities $V_i^{(k)}(x,\dot x)$  by $(x_i)_{t_k}$. 
\begin{theorem} \label{th: EL all}
Let $v_1,\ldots, v_m$ be commuting variational symmetries for the Lagrange function $L(x,\dot x)$, with the fluxes $F_1(x,\dot x), \ldots, F_m(x,\dot x)$, and with the Noether integrals $J_1(x,\dot x), \ldots,$ $J_m(x,\dot x)$. Define for $k=1,\ldots, m$:
\begin{eqnarray}    \label{eq: L k int}
L_k(x,\dot x,x_{t_k}) & = & \sum_{i=1}^N \frac{\partial L(x,\dot x)}{\partial \dot x_i} (x_i)_{t_k} -J_k(x,\dot x)\\
    \label{eq: L k flux}
                             & = & \sum_{i=1}^N \frac{\partial L(x,\dot x)}{\partial \dot x_i}\big((x_i)_{t_k}-V_i^{(k)}(x,\dot x)\big)+F_k(x,\dot x),
\end{eqnarray}
so that $L_k(x,\dot x,x_{t_k})$ coincides with $F_k(x,\dot x)$ on solutions of differential equations (\ref{eq: new flows}). 
Then for the 1-form (\ref{eq: right form}), the multi-time Euler-Lagrange equations are equivalent to a consistent system of the standard Euler-Lagrange equations $\EE_i=0$ coupled with the evolution equations (\ref{eq: new flows}).
\end{theorem}
\benpf
The multi-time Euler-Lagrange equations for the 1-form $\mathcal L$ consist of:
\begin{itemize}
\item
Equations (\ref{eq: pluri 1}), namely
\beq \label{eq: dLiddxk}
\frac{\partial L}{\partial (x_i)_{t_k}}=0, \quad \frac{\partial L_k}{\partial \dot x_i}=0, \quad \frac{\partial L_k}{\partial (x_i)_{t_\ell}}=0,\quad
i=1,\ldots, N, \quad k\neq \ell\in\{1,\ldots,m\}.
\eeq
The first and the third ones are trivially satisfied, while the second is equivalent to (\ref{eq: new flows}). Indeed, from definition (\ref{eq: L k flux}) of $L_k$, we compute the second equation in (\ref{eq: dLiddxk}) to be :
\[
\frac{\partial L_k}{\partial \dot x_i}  =  \sum_{j=1}^N \frac{\partial^2 L}{\partial \dot x_i\partial \dot x_j} \Big((x_j)_{t_k}-V_j^{(k)}\Big)
-\sum_{j=1}^N \frac{\partial L}{\partial \dot x_j}\ \frac{\partial V_j^{(k)}}{\partial \dot x_i}+\frac{\partial F_k}{\partial \dot x_i}=0,
\]
with $i=1,\ldots, N$. 
\noindent
According to Lemma \ref{lemma F}, this results in
\[
\frac{\partial L_k}{\partial \dot x_i}  = \sum_{j=1}^N \frac{\partial^2 L}{\partial \dot x_i\partial \dot x_j} \Big((x_j)_{t_k}-V_j^{(k)}\Big)=0,\qquad i=1,\ldots, N.
\]
Due to the non-degeneracy condition (\ref{eq: nondeg}), this is equivalent to equations (\ref{eq: new flows}). 
\item
Equations (\ref{eq: pluri 2}) read:
\[
\frac{\partial L}{\partial \dot x_i}=\frac{\partial L_1}{\partial (x_i)_{t_1}}=\ldots=\frac{\partial L_m}{\partial (x_i)_{t_m}}, \qquad  i=1,\ldots, N,
\]
and are automatically satisfied by construction. 
\item
Equations (\ref{eq: pluri 3}) consist of the standard Euler-Lagrange equations $\EE_i=0$ given in (\ref{eq: EL 1}), and of the Euler-Lagrange equations associated with $L_k$:
\beq \label{eq: EL k}
\EE^{(k)}_i = \frac{\pa L_k(x, \dot x,x_{t_k})}{\pa x_i}- \D_{t_k} p_i(x, \dot x)=0,
\eeq
with $i=1,\dots,N, k=1,\ldots,m$.
The latter will be shown to be differential consequences of $\EE_i=0$ and equations (\ref{eq: new flows}), see Proposition \ref{lemma var eq 1} below.
\end{itemize}
This finishes the proof of Theorem \ref{th: EL all}.
\endpf

{\bf Remark.} It is a highly nontrivial feature of the pluri-Lagrangian theory that the multi-time Euler-Lagrange equations include the  evolutionary first order differential equations (\ref{eq: new flows}). 

\begin{proposition}\label{lemma var eq 1}
For $L_k(x,\dot x,x_{t_k})$ defined as in (\ref{eq: L k flux}), Euler-Lagrange equations (\ref{eq: EL k}) are differential consequences of $\EE_i=0$ and equations (\ref{eq: new flows}). 
\end{proposition}
\benpf 
Both parts of equation $\EE_i^{(k)}=0$, that is, of equation $\D_{t_k}p_i=\partial L_k/\partial x_i$, written in length, read:
\beq \label{eq: dpidtk}
\D_{t_k}p_i=\sum_{j=1}^N\frac{\partial^2 L}{\partial \dot x_i\partial x_j} (x_j)_{t_k}+\sum_{j=1}^N\frac{\partial^2 L}{\partial \dot x_i\partial \dot x_j} (\dot x_j)_{t_k}\ ,
\eeq
and
\beq \label{eq: dLkdxi}
\frac{\partial L_k}{\partial x_i}=\sum_{j=1}^N\frac{\partial^2 L}{\partial x_i\partial \dot x_j}\big((x_j)_{t_k}-V_j^{(k)}\big)-\sum_{j=1}^N \frac{\partial L}{\partial \dot x_j}\ \frac{\partial V_j^{(k)}}{\partial x_i}+\frac{\partial F_k}{\partial x_i}\ .
\eeq
On solutions of $(x_j)_{t_k}=V_j^{(k)}$ they reduce respectively to 
\beq \label{eq: dpidtk restr}
\D_{t_k}p_i=\sum_{j=1}^N\frac{\partial^2 L}{\partial \dot x_i\partial x_j} V_i^{(k)}+\sum_{j=1}^N\frac{\partial^2 L}{\partial \dot x_i\partial \dot x_j} \big(\D_t V_j^{(k)}\big),
\eeq
and 
\beq \label{eq: dLkdxi restr}
\frac{\partial L_k}{\partial x_i}=-\sum_{j=1}^N \frac{\partial L}{\partial \dot x_j}\ \frac{\partial V_j^{(k)}}{\partial x_i}+\frac{\partial F_k}{\partial x_i}\ .
\eeq
We transform expression (\ref{eq: dpidtk restr})  as follows. First, we integrate by parts and use equation (\ref{eq: var symmetries}) for variational symmetries:
\begin{eqnarray*}
\D_{t_k}p_i & = & \frac{\partial}{\partial \dot x_i}\left(\sum_{j=1}^N\frac{\partial L}{\partial x_j} V_j^{(k)}
+\sum_{j=1}^N\frac{\partial L}{\partial \dot x_j} \big(\D_t V_j^{(k)}\big)\right)\nn \\
&&
-\sum_{j=1}^N\frac{\partial L}{\partial x_j} \ \frac{\partial V_j^{(k)}}{\partial \dot x_i}
-\sum_{j=1}^N\frac{\partial L}{\partial \dot x_j} \ \frac{\partial}{\partial \dot x_i}\big(\D_t V_j^{(k)}\big)\\
& = & \frac{\partial}{\partial \dot x_i}(\D_tF_k)-\sum_{j=1}^N\frac{\partial L}{\partial x_j}\ \frac{\partial V_j^{(k)}}{\partial \dot x_i}
 -\sum_{j=1}^N\frac{\partial L}{\partial \dot x_j} \ \frac{\partial}{\partial \dot x_i}\big(\D_t V_j^{(k)}\big) .
\end{eqnarray*}
In the first sum of the last line, we make use of Euler-Lagrange equations (\ref{eq: EL 1}), to get:
\[
\D_{t_k}p_i=\frac{\partial}{\partial \dot x_i}(\D_tF_k)-\sum_{j=1}^N\EE_i\ \frac{\partial V_j^{(k)}}{\partial \dot x_i}
-\sum_{j=1}^N \D_t\left(\frac{\partial L}{\partial \dot x_j}\right)\ \frac{\partial V_j^{(k)}}{\partial \dot x_i}
-\sum_{j=1}^N\frac{\partial L}{\partial \dot x_j} \ \frac{\partial}{\partial \dot x_i}\big(\D_t V_j^{(k)}\big) .
\] 
Integrating by parts once again, we get:
\begin{eqnarray} \label{eq: lemma 2 aux 2}
\D_{t_k}p_i & = & \frac{\partial}{\partial \dot x_i}(\D_tF_k)-\sum_{j=1}^N\EE_i\ \frac{\partial V_j^{(k)}}{\partial \dot x_i}
-\D_t\! \left(\sum_{j=1}^N\frac{\partial L}{\partial \dot x_j}\ \frac{\partial V_j^{(k)}}{\partial \dot x_i}\right) \nonumber \\
  & & + \sum_{j=1}^N\frac{\partial L}{\partial \dot x_j}\left( \D_t\bigg(\frac{\partial V_j^{(k)}}{\partial \dot x_i}\bigg)-
  \frac{\partial}{\partial \dot x_i}\big(\D_t V_j^{(k)}\big)\right).
\end{eqnarray}
But for any function $f(x,\dot x)$, we have the following identity:
\begin{eqnarray}
\lefteqn{\D_t\left(\frac{\partial f}{\partial \dot x_i}\right)-\frac{\partial }{\partial \dot x_i}(\D_t f)} \nonumber \\
 & = & 
\sum_{j=1}^N\frac{\partial^2 f}{\partial \dot x_i \partial x_j} \dot x_j+\sum_{j=1}^N\frac{\partial^2 f}{\partial \dot x_i \partial \dot x_j} \ddot x_j
-\frac{\partial}{\partial \dot x_i}\left(\sum_{j=1}^N\frac{\partial f}{\partial x_j} \dot x_j+\sum_{j=1}^N\frac{\partial f}{\partial \dot x_j} \ddot x_j\right)\nonumber \\
 & = & -\frac{\partial f}{\partial x_i}. \nn
\end{eqnarray}
Using this identity twice (for $F_k$ and for $V_j^{(k)}$), we put (\ref{eq: lemma 2 aux 2}) in the form
\[
\D_{t_k}p_i =-\sum_{j=1}^N\EE_i\ \frac{\partial V_j^{(k)}}{\partial \dot x_i}+\D_t
\left(\frac{\partial F_k}{\partial \dot x_i}-\sum_{j=1}^N\frac{\partial L}{\partial \dot x_j}\ \frac{\partial V_j^{(k)}}{\partial \dot x_i}\right) 
+\frac{\partial F_k}{\partial x_i} - \sum_{j=1}^N\frac{\partial L}{\partial \dot x_j}\ \frac{\partial V_j^{(k)}}{\partial x_i}.
\]
According to Lemma \ref{lemma F}, this equals to
\[
\D_{t_k}p_i =-\sum_{j=1}^N\EE_i\ \frac{\partial V_j^{(k)}}{\partial \dot x_i}
+\frac{\partial F_k}{\partial x_i} - \sum_{j=1}^N\frac{\partial L}{\partial \dot x_j}\ \frac{\partial V_j^{(k)}}{\partial x_i},
\]
which on solutions of $\EE_i=0$ coincides with the right-hand side of (\ref{eq: dLkdxi restr}). 
\endpf

\section{Almost closedness of the pluri-Lagrangian 1-form}
\label{Sect almost closed}

As a corollary of Theorem \ref{Th almost closed}, we can conclude that 
\begin{align}
 \D_{t_k} L(x,\dot x)- \D_{t} L_k(x,\dot x,x_{t_k}) & =0, \label{eq: DkL-DtLk}\\ 
 \D_{t_k} L_\ell(x,\dot x,x_{t_\ell})- \D_{t_\ell} L_k(x,\dot x,x_{t_k}) & =c_{k\ell},  \label{eq: DkLl-DlLk}
\end{align} 
on solutions of the Euler-Lagrange equations $\EE_i=0$ and of the evolution equations of variational symmetries $(x_i)_{t_k}=V^{(k)}_i$. However, a more detailed information is available.

\begin{proposition} \label{Th: DkLl-DlLk}
The following identities hold:
\begin{eqnarray} 
\D_{t_k} L(x,\dot x)-\D_{t} L_k(x,\dot x,x_{t_k}) & = & \sum_{i=1}^N\big((x_i)_{t_k}-V_i^{(k)} \big) \EE_i\ , 
\label{eq: DkL-DtLk full} \\
\D_{t_k} L_\ell(x,\dot x,x_{t_\ell})- \D_{t_\ell} L_k(x,\dot x,x_{t_k}) & = & c_{k\ell}+
\sum_{i=1}^N \left(\big((x_i)_{t_\ell}-V_i^{(\ell)}\big)\EE_i^{(k)}-\big((x_i)_{t_k}-V_i^{(k)}\big)\EE_i^{(\ell)} \right)\nonumber \\
\lefteqn{\!\!\!\!\!\!\!\!\!\!\!\!\!\!\!\!+\sum_{i=1}^N\sum_{j=1}^N\bigg(\frac{\pa^2 L}{\pa x_i\pa \dot x_j}-\frac{\pa^2 L}{\pa x_j\pa \dot x_i}\bigg)
\big((x_i)_{t_\ell}-V_i^{(\ell)}\big)\big((x_j)_{t_k}-V_j^{(k)}\big),} \label{eq: DkLl-DlLk full}
\end{eqnarray}
where $c_{k\ell}$ are constants analogous to those from Proposition \ref{prop comm symm}. In particular, equations (\ref{eq: DkL-DtLk}), (\ref{eq: DkLl-DlLk}) are satisfied as soon as any $m$ out of $m+1$ systems $\EE_i=0$ and $(x_i)_{t_k}=V^{(k)}_i$ $(k=1,\ldots,m)$ are satisfied.
\end{proposition}

\benpf 
For the proof of equation (\ref{eq: DkL-DtLk full}), we compute:
\[
\D_{t_k} L(x,\dot x)-\D_{v_k}L(x,\dot x)  =  
\sum_{i=1}^N \big((x_i)_{t_k}-V_i^{(k)}\big)\frac{\pa L}{ \pa x_i} 
 +\sum_{i=1}^N \D_t\big((x_i)_{t_k}-V_i^{(k)}\big)\frac{\pa L}{ \pa \dot x_i},
\]
and
\[
\D_{t} L_k(x,\dot x,x_{t_k})  =  \sum_{i=1}^N \D_t\big((x_i)_{t_k}-V_i^{(k)}\big)\frac{\pa L}{ \pa \dot x_i}
 +\sum_{i=1}^N \big((x_i)_{t_k}-V_i^{(k)}\big)\D_t\bigg(\frac{\pa L}{ \pa \dot x_i}\bigg)+\D_{t} F_k.
\]
Taking the difference and using (\ref{eq: EL 1}), we find:
$$
\D_{t_k} L(x,\dot x)-\D_{t} L_k(x,\dot x,x_{t_k})=\D_{v_k}L-\D_t F_k+\sum_{i=1}^N \big((x_i)_{t_k}-V_i^{(k)}\big)\EE_i\ .
$$
Taking into account equation (\ref{eq: var symmetries}), we arrive at (\ref{eq: DkL-DtLk full}).

For the proof of equation (\ref{eq: DkLl-DlLk full}), we compute:
\begin{eqnarray*}
\D_{t_k} L_\ell(x,\dot x,x_{t_\ell}) & = & 
\sum_{i=1}^N \big(\D_{t_k} p_i \big)\big((x_i)_{t_\ell}-V_i^{(\ell)}\big)  +\sum_{i=1}^N p_i (x_i)_{t_kt_\ell} -\sum_{i=1}^N  p_i  \big(\D_{t_k}V_i^{(\ell)}\big)+\D_{t_k}  F_\ell \\
  & = & \sum_{i=1}^N \left(-\EE_i^{(k)}+\frac{\partial L_k}{\partial x_i}\right)\big((x_i)_{t_\ell}-V_i^{(\ell)}\big) \nn \\
  && +\sum_{i=1}^N p_i (x_i)_{t_kt_\ell} -\sum_{i=1}^N  p_i  \big(\D_{t_k}V_i^{(\ell)}\big)+\D_{t_k}  F_\ell. 
  \end{eqnarray*}
Here
\[
\D_{t_k}F_\ell=D_{v_k}F_\ell+\sum_{i=1}^N \big((x_i)_{t_k}-V_i^{(k)}\big)\frac{\partial F_\ell}{\partial x_i}+\sum_{i=1}^N \D_t\big((x_i)_{t_k}-V_i^{(k)}\big)\frac{\partial F_\ell}{\partial \dot x_i}\ .
\]  
Thus, we find:
\begin{eqnarray}  \label{eq: proof prop 2 aux}
\lefteqn{\D_{t_k} L_\ell(x,\dot x,x_{t_\ell}) -\D_\ell L_k(x,\dot x,x_{t_k})  =  \D_{v_k}F_\ell-\D_{v_\ell}F_k}  \nonumber \\
  & & +\sum_{i=1}^N \left(\big((x_i)_{t_k}-V_i^{(k)}\big)\EE_i^{(\ell)}-\big((x_i)_{t_\ell}-V_i^{(\ell)}\big)\EE_i^{(k)}\right)  \nonumber \\ 
  & & +\sum_{i=1}^N \frac{\partial L_k}{\partial x_i}\big((x_i)_{t_\ell}-V_i^{(\ell)}\big) 
         -\sum_{i=1}^N \big((x_i)_{t_\ell}-V_i^{(\ell)}\big)\frac{\partial F_k}{\partial x_i}
         -\sum_{i=1}^N \D_t\big((x_i)_{t_\ell}-V_i^{(\ell)}\big)\frac{\partial F_k}{\partial \dot x_i}   \nonumber \\
  & & -\sum_{i=1}^N \frac{\partial L_\ell}{\partial x_i}\big((x_i)_{t_k}-V_i^{(k)}\big)
        +\sum_{i=1}^N \big((x_i)_{t_k}-V_i^{(k)}\big)\frac{\partial F_\ell}{\partial x_i}
        +\sum_{i=1}^N \D_t\big((x_i)_{t_k}-V_i^{(k)}\big)\frac{\partial F_\ell}{\partial \dot x_i}  \nonumber \\
& &   +\sum_{i=1}^N  p_i  \big(\D_{t_\ell}V_i^{(k)}-\D_{t_k}V_i^{(\ell)}\big).
\end{eqnarray}
Here the last line (with the index $i$ replaced by $j$) is transformed as follows:
\begin{eqnarray*}
\lefteqn{\sum_{j=1}^N p_j\big(\D_{t_\ell}V_j^{(k)}-\D_{t_k}V_j^{(\ell)}\big)\ = \ \sum_{j=1}^N p_j\big(\D_{v_\ell}V_j^{(k)}-\D_{v_k}V_j^{(\ell)}\big)}\\
& & +\sum_{j=1}^N p_j\sum_{i=1}^N \big((x_i)_{t_\ell}-V_i^{(\ell)}\big)\frac{\partial V_j^{(k)}}{\partial x_i}
       +\sum_{j=1}^N p_j\sum_{i=1}^N \D_t\big((x_i)_{t_\ell}-V_i^{(\ell)}\big)\frac{\partial V_j^{(k)}}{\partial \dot x_i}\\
& &  -\sum_{j=1}^N p_j\sum_{i=1}^N \big((x_i)_{t_k}-V_i^{(k)}\big)\frac{\partial V_j^{(\ell)}}{\partial x_i}
       -\sum_{j=1}^N p_j\sum_{i=1}^N \D_t\big((x_i)_{t_k}-V_i^{(k)}\big)\frac{\partial V_j^{(\ell)}}{\partial \dot x_i}\ .
\end{eqnarray*}
Since the symmetries $v_k$ and $v_\ell$ commute, we have from Proposition \ref{prop comm symm}:
\[
\D_{v_k} F_\ell-\D_{v_\ell} F_k=c_{k\ell}+\sum_{j=1}^Np_i\left(\D_{v_k} V_j^{(\ell)}-\D_{v_\ell} V_j^{(k)}\right).
\]
Collecting everything together and using Lemma \ref{lemma F}, we see that in (\ref{eq: proof prop 2 aux}) all terms with $\D_t\big((x_i)_{t_k}-V_i^{(k)}\big)$ cancel away, and we obtain:
\begin{eqnarray*}
\lefteqn{\D_{t_k} L_\ell(x,\dot x,x_{t_\ell}) -\D_\ell L_k(x,\dot x,x_{t_k})  =  c_{k\ell}
   +\sum_{i=1}^N \left(\big((x_i)_{t_k}-V_i^{(k)}\big)\EE_i^{(\ell)}-\big((x_i)_{t_\ell}-V_i^{(\ell)}\big)\EE_i^{(k)}\right)}\\ 
  & & +\sum_{i=1}^N \frac{\partial L_k}{\partial x_i}\big((x_i)_{t_\ell}-V_i^{(\ell)}\big) 
         -\sum_{i=1}^N \big((x_i)_{t_\ell}-V_i^{(\ell)}\big)\frac{\partial F_k}{\partial x_i}
          +\sum_{j=1}^N p_j\sum_{i=1}^N \big((x_i)_{t_\ell}-V_i^{(\ell)}\big)\frac{\partial V_j^{(k)}}{\partial x_i}\\
  & & -\sum_{i=1}^N \frac{\partial L_\ell}{\partial x_i}\big((x_i)_{t_k}-V_i^{(k)}\big)
        +\sum_{i=1}^N \big((x_i)_{t_k}-V_i^{(k)}\big)\frac{\partial F_\ell}{\partial x_i}
         -\sum_{j=1}^N p_j\sum_{i=1}^N \big((x_i)_{t_k}-V_i^{(k)}\big)\frac{\partial V_j^{(\ell)}}{\partial x_i}\  .     
\end{eqnarray*}  
It remains to substitute expressions for $\partial L_k/\partial x_i$ from (\ref{eq: dLkdxi}).
\endpf

\section{Pluri-Lagrangian problems in the phase space}
\label{Sect phase space}

We recall that motions of mechanical systems can be described as extremals of a more general variational principle that
the Hamilton's principle of the least action, namely that they are critical points of the following {\em action functional in the phase space}:
\begin{equation}\label{eq: action phase space}
S[x,p]=\int_{t_0}^{t_1} \Lambda(x(t),p(t),\dot x(t))\dd t
\end{equation}
where $(x,p):[t_0,t_1]\to T^*\mathcal X$ is an arbitrary curve in the phase space with the fixed values of $x(t_0)$ and $x(t_1)$, and the ``Lagrange function'' $\Lambda: T(T^*\mathcal X)\to \mathbb R$ is given by 
\begin{equation}\label{eq: Lagr phase space}
\Lambda(x,p,\dot x)=\sum_{i=1}^N p_i\dot x_i-H(x,p).
\end{equation}
Note that this ``Lagrange function'' is degenerate in two ways. First, it depends on the ``velocities'' $(\dot x,\dot p)$ linearly. Second, it actually does not depend on one half of velocities, namely on $\dot p$. This last feature forces us to consider admissible variations fixing $x(t_0)$ and $x(t_1)$, but not $p(t_0)$ and $p(t_1)$. However, Euler-Lagrange equations for the ``Lagrange function'' $\Lambda(x,p,\dot x)$ are computed in the standard way, using
$$
\frac{\partial \Lambda}{\partial x_i}=-\frac{\partial H}{\partial x_i}, \quad \frac{\partial \Lambda}{\partial p_i}=\dot{x_i}-\frac{\partial H}{\partial p_i}, \quad 
\frac{\partial \Lambda}{\partial \dot x_i}=p_i, \quad \frac{\partial \Lambda}{\partial \dot p_i}=0,
$$ 
and read:
\begin{equation}\label{eq: Ham from EL phase space}
-\frac{\partial H}{\partial x_i}-\D_t p_i=0, \quad \dot{x_i}-\frac{\partial H}{\partial p_i}=0,
\end{equation}
which coincides with the Hamiltonian equation of motion (\ref{eq: Ham eq}). An important insight is that the usual Lagrange function $L(x,\dot x)$ is the critical value of $\Lambda(x,p,\dot x)$ with respect to $p$, which is achieved at $p$ given by the Legendre transformation $p_i=\partial L(x,\dot x)/\partial x_i$. This explains why the critical curves of the action (\ref{lp}) in the configuration space $\mathcal X$, in the class of variations fixing $x(t_0)$ and $x(t_1)$, are obtained by substitution $p_i=\partial L(x,\dot x)/\partial x_i$ from the critical curves of the action (\ref{eq: action phase space}) in the phase space $T^*\mathcal X$, in the class of variations fixing $x(t_0)$ and $x(t_1)$, but not $p(t_0)$ and $p(t_1)$.

We now observe that the coefficients $L_k(x,\dot x,x_{t_k})$ given by equation (\ref{eq: L k int}) are, in a similar way, the values of the functions
\begin{equation}\label{eq: Lk phase space}
\Lambda_k(x,p,x_{t_k})=\sum_{i=1}^N p_i (x_i)_{t_k}-H_k(x,p),
\end{equation}
evaluated at the point $p$ given by the Legendre transformation $p_i=\partial L(x,\dot x)/\partial x_i$. Here $H_k: T^*\mathcal X\to\mathbb R$ are the Hamilton functions related, via the Legendre transformation, to the Noether integrals $J_k: T\mathcal X\to \mathbb R$ of the commuting variational symmetries $v_k$, as in Section \ref{Sect Ham}. Comparing formulas (\ref{eq: Lagr phase space}) and (\ref{eq: Lk phase space}), we see that the time $t$ associated with the Lagrange function $L$ is on absolutely the same footing as the times $t_k$ associated with the commuting variational symmetries.  

This leads us naturally to consider the following {\em pluri-Lagrangian problem in the phase space} (which we formulate without choosing one time playing a special role, like we did in Section \ref{Sect pluri}). Consider the 1-form 
\begin{equation}\label{eq: form phase space}
\mathcal L=\Lambda_1(x,p,x_{t_1})\dd t_1+\ldots+\Lambda_m(x,p,x_{t_m})\dd t_m,
\end{equation}
whose coefficients depend on (the first jet of) a function $(x,p):\mathbb R^{m}\to T^*\mathcal X$ according to formulas (\ref{eq: Lk phase space}).
Find functions $(x,p):\mathbb R^{m}\to T^*\mathcal X$ which deliver critical points to  functionals
$$
S_\Gamma=\int_\Gamma \mathcal L
$$
for any curve $\Gamma:[0,1]\to\mathbb R^{m}$ in the multi-dimensional space (for fixed values of $x$ at the endpoints of $\Gamma$).

\begin{theorem} \label{Th multi-time EL eqs phase space} {\bf (Multi-time Euler-Lagrange equations in the phase space)}
A function $(x,p):\R^{m}\to T^*\mathcal X$ solves the pluri-Lagrangian problem for the $1$-form (\ref{eq: form phase space})  with coefficients (\ref{eq: Lk phase space}), where $H_1, \ldots , H_m:T^*\mathcal X\to\mathbb R$ are some functions, if and only if it satisfies Hamiltonian equations for the Hamilton functions $H_1,\ldots, H_m$:
\beq\label{eq: Ham eq k}
(x_i)_{t_k}=\frac{\partial H_k}{\partial p_i}, \quad (p_i)_{t_k}=-\frac{\partial H_k}{\partial x_i},\quad  k=1,\dots, m.
\eeq
This system is compatible if and only if these $m$ Hamiltonian flows pairwise commute.
\end{theorem}
\begin{proof}
One easily computes the multi-time Euler-Lagrange equations for the $1$-form (\ref{eq: form phase space})  with coefficients (\ref{eq: Lk phase space}):
\begin{itemize}
\item equations (\ref{eq: pluri 1}) are satisfied trivially by construction;
\item equations (\ref{eq: pluri 2}) are also satisfied by construction, since $\partial \Lambda_k/\partial (x_i)_{t_k}=p_i$ for all $k$;
\item equations (\ref{eq: pluri 3}) read: $\D_{t_k}p_i=-\partial H_k/\partial x_i$ and $0=(x_i)_{t_k} -\partial H_k/\partial x_i$. These are derived literally as equations (\ref{eq: Ham from EL phase space}).
\end{itemize}
The last statement of the theorem is obvious.
\end{proof}

Thus, we established the property of joint solutions of a system of pairwise commuting Hamiltonian flows to be critical for the action functionals in the phase space defined along any curve in the multi-dimensional time.

\section{Examples}
\label{Sec examples}

We illustrate the results of the present paper with a couple of well-known examples.

\subsection{Kepler system}

We start with the Kepler system, an example with a ``hidden symmetry'' described by the so called Runge-Lenz vector, a non-obvious integral whose presence ensures that the system is super-integrable, or integrable in the non-commutative sense, with one-dimensional invariant tori (which in this case are just periodic orbits foliating the phase space).  

\begin{itemize}

\item The Lagrange function of the Kepler system (with unit mass) reads:
\beq
L(x,\dot x)=\frac12 ( \dot x_1^2 +  \dot x_2^2 +\dot x_3^2)- \frac{\alpha}{\| x \| },
\label{lag kepler}
\eeq
where $\alpha>0$ is a constant and $\| x \| =  (x_1^2 +  x_2^2 +x_3^2)^{1/2}$. The corresponding Euler-Lagrange equations are
\beq
\EE_i=\frac{\alpha x_i}{\| x\|^3}+ \ddot x_i=0,
\qquad i=1,2,3. \label{EL kepler}
\eeq
They admit the energy integral
\beq
J(x,\dot x)=\frac12 ( \dot x_1^2 +  \dot x_2^2 +\dot x_3^2)+ \frac{\alpha}{\| x \| },
\label{energy kepler}
\eeq
as well as three angular momenta integrals $M_i=x_j\dot x_k-x_k\dot x_j$. The latter can be considered as Noether integrals due to the point symmetries $x_j\partial/\partial x_k-x_k\partial/\partial x_j$.

\item 
The Kepler system admits a variational symmetry:
\beq
v_1=V^{(1)}_1(x, \dot x)\frac{\pa}{\pa x_1}+V^{(1)}_2(x, \dot x)\frac{\pa}{\pa x_2}+ V^{(1)}_3(x, \dot x) \frac{\pa}{\pa x_3},\label{symm kepler}
\eeq
with
\bea
V^{(1)}_1(x, \dot x)&=&x_2 \dot x_2 + x_3 \dot x_3,
\nn \\
V^{(1)}_2(x, \dot x)&=&x_2 \dot x_1 -2 x_1 \dot x_2,
\nn \\
V^{(1)}_3(x, \dot x)&=&x_3 \dot x_1 -2  x_1 \dot x_3.
\nn
\eea
Indeed, a straightforward computation confirms that
$$
\D_{v_1} L (x, \dot x)- \D_{t} F_1(x, \dot x)=0, 
$$
where the flux $F_1$ is given by
$$
F_1(x, \dot x) =
\dot x_1(x_2\dot x_2 + x_3\dot x_3) -x_1 ( \dot x_2^2 + \dot x_3^2 )- \frac{\alpha x_1}{\| x\|}.
$$
The corresponding Noether integral is 
\beq
J_1(x, \dot x) =
\dot x_1 (x_2\dot x_2 + x_3\dot x_3) -x_1 ( \dot x_2^2 + \dot x_3^2 )+\frac{\alpha x_1}{\| x\|}.
\label{J kepler}
\eeq

\item To the variational symmetry $v_1$ given by (\ref{symm kepler}) we can associate a flow with the
time $t_1$:
\bea
(x_1)_{t_1}&=& x_2 \dot x_2 + x_3 \dot x_3,  \nn \\
(x_2)_{t_1}&=&x_2 \dot x_1 -2 x_1 \dot x_2, 
\label{flow1 kepler} \\
(x_3)_{t_1}&=& x_3\dot x_1 -2  x_1 \dot x_3. \nn
\eea
Set
$$
L_1 (x,\dot x,x_{t_1})=\sum_{i=1}^3 \dot x_i (x_i)_{t_1} -J_1(x,\dot x).
$$
Then the Euler-Lagrange equations associated with $L_1$ (see (\ref{eq: EL k})) are computed to be
\bea
\EE_1^{(1)}&=& \dot x_2^2 + \dot x_3^2 -
\frac{\alpha (x_2^2+x_3^2)}{\|x\|^3}-(\dot x_1)_{t_1}=0, \nn \\
\EE_2^{(1)}&=& 
- \dot x_1 \dot x_2 +\frac{\alpha x_1 x_2}{\|x\|^3}
-(\dot x_2)_{t_1}=0, \label{E kepler} \\
\EE_3^{(1)}&=& 
 -\dot x_1 \dot x_3 +\frac{\alpha x_1 x_3}{\|x\|^3}
-(\dot x_3)_{t_1}=0. \nn
\eea
As ensured by Proposition \ref{lemma var eq 1}, equations (\ref{E kepler})  are differential consequences of (\ref{EL kepler}) and 
(\ref{flow1 kepler}). Indeed, differentiating (\ref{flow1 kepler}) with respect to $t$ and then substituting 
$\ddot x_i=-\alpha x_i/\| x\|^3$ (coming from (\ref{EL kepler})) in the resulting expression, we recover (\ref{E kepler}).

\item Actually, Kepler system possesses two further variational symmetries, say $v_2$ and $v_3$, which are obtained by permuting coordinates $x_1,x_2,x_3$ in (\ref{symm kepler}). These permutations give also the fluxes and Noether integrals corresponding to symmetries $v_2$ and $v_3$. One can consider the integrals $J_1,J_2,J_3$ as the components of the famous Runge-Lenz vector:
$$
J(x, \dot x)= \dot x \times (x \times \dot x)+ \frac{\alpha x}{\| x\|}.
$$
However, symmetries $v_1,v_2,v_3$ do not commute on solutions of Euler-Lagrange equations (\ref{EL kepler}). For example, one finds
\bea
\D_{v_1} V_1^{(2)}(x, \dot x) -\D_{v_2} V_1^{(1)}(x, \dot x) &=&-\ddot x_2 ( 2 x_1^2 - 2 x_2^2 -x_3^2 ) - x_2 x_3 \ddot x_3\nn \\
&&  - x_2 (3 \dot x_1^2 +\dot x_2^2+ \dot x_3^2) + 2 x_1 \dot x_1 \dot x_2, \nn
\eea
which does not vanish modulo $\mathcal E_i=0$.
\end{itemize}

\subsection{Toda lattice}

\begin{itemize}

\item The Lagrangian of the Toda lattice reads
\beq
L(x,\dot x)=\sum_{i=1}^N \left( \frac12
\dot x_i^2 - \E^{x_{i+1}-x_i} \right).
\label{Lag toda}
\eeq
One usually imposes one of the two types of boundary conditions: periodic, 
$x_0\equiv x_N$,  $x_{N+1}\equiv x_1$, or open-end, $x_{0}=+\infty$, $x_{N+1}=-\infty$.
 The corresponding Euler-Lagrange equations are
\beq
\EE_i=\E^{x_{i+1}-x_i}- \E^{x_{i}-x_{i-1}}-\ddot x_i=0,
\qquad i=1,\dots,N. \label{EL toda}
\eeq

\item The Toda lattice possesses $N$ commuting variational symmetries. The two simplest ones are:
\beq
v_1= \sum_{i=1}^N V_i ^{(1)}(x, \dot x)\frac{\pa}{\pa x_i}, \qquad
v_2= \sum_{i=1}^N V_i ^{(2)}(x, \dot x)\frac{\pa}{\pa x_i}, \label{symm Toda}
\eeq
with characteristics
\bea
V_i ^{(1)}(x, \dot x)&=&
\dot x_i^2 +\E^{x_{i+1}-x_i}+\E^{x_{i}-x_{i-1}}, \nn \\
V_i ^{(2)}(x, \dot x)&=&
\dot x_i^3 +
( \dot x_{i+1}+2 \dot x_i)\E^{x_{i+1}-x_i}+
(2 \dot x_i+ \dot x_{i-1})\E^{x_{i}-x_{i-1}}.\nn
\eea
Indeed, a simple computation confirms that
$$
\D_{v_k} L (x, \dot x)- \D_{t} F_k(x, \dot x)=0, \qquad
k=1,2,
$$
where the fluxes $F_1$ and $F_2$ are given by
\bea
F_1(x, \dot x) &=&\frac23\sum_{i=1}^N 
\dot x_i^3, \nn \\
F_2(x, \dot x) &=&\sum_{i=1}^N \left( \frac34 \dot x_i^4
+( \dot x_i^2 +  \dot x_i  \dot x_{i+1}+ \dot x_{i+1}^2 )\E^{x_{i+1}-x_i} \right) \nn \\
&&-\sum_{i=1}^N \left( \frac12 \E^{2(x_{i+1}-x_i)}-\E^{x_{i+2}-x_i} \right). \nn
\eea
The Noether integrals corresponding to
the characteristics $V_i ^{(1)}$ and $V_i ^{(2)}$
are:
\bea
J_1(x, \dot x)  &=&\sum_{i=1}^N
\left(\frac13 \dot x_i^3 + (\dot x_i + \dot x_{i+1})\E^{x_{i+1}-x_i} \right), \label{J1 toda} \\
J_2(x, \dot x) &=&\sum_{i=1}^N \left( \frac14 \dot x_i^4
+(\dot x_i^2 + \dot x_i \dot x_{i+1}+\dot x_{i+1}^2 )\E^{x_{i+1}-x_i} \right)\nn \\
&&+ \sum_{i=1}^N \left(\frac12 \E^{2(x_{i+1}-x_i)}+\E^{x_{i+2}-x_i} \right). \label{J2 toda}
\eea
The two symmetries $v_1$ and $v_2$ commute on solutions of Euler-Lagrange equations (\ref{EL toda}), since
$$
\D_{v_1} V_i^{(2)}(x, \dot x) -\D_{v_2} V_i^{(1)}(x, \dot x)  = \sum_{j=1}^N  r_{ij}(x, \dot x) \EE_j\ , \quad i=1,\ldots, N,
$$
with
$$
r_{ij}(x, \dot x) =-2 (\dot x_{i+1}-\dot x_i ) \E^{x_{i+1}-x_i} \delta_{i+1,j}
+2 (\dot x_i - \dot x_{i-1}) \E^{x_{i}-x_{i-1}} \delta_{i-1,j}.
$$

\item The two flows corresponding to the commuting variational symmetries $v_1$ and $v_2$ are:
\bea
(x_i )_{t_1}&=&
\dot x_i^2 +\E^{x_{i+1}-x_i}+\E^{x_{i}-x_{i-1}}, \label{flow1} \\
(x_i)_{t_2}&=&
\dot x_i^3 +
(2 \dot x_i+ \dot x_{i+1})\E^{x_{i+1}-x_i}+
(2 \dot x_i+ \dot x_{i-1})\E^{x_{i}-x_{i-1}},\label{flow2}
\eea
with $i=1,\dots,N$.  For $k=1,2$, we define
$$
L_k (x,\dot x,x_{t_k})=\sum_{i=1}^N \dot x_i (x_i)_{t_k} -J_k(x,\dot x),
$$
with $J_1$, $J_2$ from (\ref{J1 toda}), (\ref{J2 toda}).
Euler-Lagrange equations (\ref{eq: EL k}) associated with
$L_1$ and $L_2$ read respectively:
\beq
\EE_i^{(1)}=(\dot x_i+\dot x_{i+1})\E^{x_{i+1}-x_{i}}-
(\dot x_i+\dot x_{i-1})\E^{x_i-x_{i-1}} - (\dot x_i)_{t_1}=0,
\label{E1}
\eeq
and
\bea
\EE_i^{(2)}&=&(\dot x_i^2 + \dot x_i\dot x_{i+1}+\dot x_{i+1}^2 )\E^{x_{i+1}-x_i}  -
(\dot x_{i-1}^2 + \dot x_{i-1}\dot x_i+\dot x_i^2 )\E^{x_{i}-x_{i-1}}\nn \\
&& +\E^{2(x_{i+1}-x_{i})}-\E^{2(x_{i}-x_{i-1})}+\E^{x_{i+2}-x_{i}}-\E^{x_{i}-x_{i-2}}- (\dot x_i)_{t_2}=0, \label{E2}
\eea
with $i=1,\dots,N$. 
As stated in Proposition \ref{lemma var eq 1}, equations (\ref{E1}) and (\ref{E2}) are differential consequences of (\ref{EL toda}) and 
(\ref{flow1})--(\ref{flow2}). Indeed, to derive (\ref{E1}), (\ref{E2}), one can differentiate (\ref{flow1}), resp. (\ref{flow2}), with respect to $t$ and then substitute $\ddot x_i=\E^{x_{i+1}-x_i}- \E^{x_{i}-x_{i-1}}$ (coming from (\ref{EL toda})) into the resulting expressions. 

\end{itemize}

\section{Conclusions}

In hindsight, the Hamiltonian picture of the pluri-Lagrangian structure on the phase space, generated by commuting Hamiltonian flows, as stated in Theorem \ref{Th multi-time EL eqs phase space}, looks to be the most simple and natural and to provide a genuine explanation to the corresponding results in the configuration space, as given in Sections \ref{Sect pluri config space}, \ref{Sect almost closed}. However, the purely Lagrangian point of view seems to be more universal and directly applicable in many contexts, where the Hamiltonian point of view is tricky or just unavailable, like discrete time and/or higher dimensional problems described by partial differential or partial difference equations \cite{BPS1,BPS2,BPS3,S2,S3}.

\section*{Acknowledgments}
This research is supported by the DFG Collaborative Research Center TRR 109 ``Discretization in Geometry and Dynamics''. We are grateful to Nicolai Reshetikhin and to Nikolai Dimitrov for useful discussions.


\end{document}